\documentclass[
  10pt,
  twocolumn,
  english,
  aps,
  pra,
  longbibliography,
  superscriptaddress,
  floatfix,
  groupedaddress
]{revtex4-1}

\usepackage{amsmath}
\usepackage{amssymb}
\usepackage{mathtools}
\usepackage{graphicx}
\usepackage{physics}
\usepackage[dvipsnames]{xcolor}
\usepackage{amsthm}
\usepackage{float}

\usepackage[english]{babel}

\usepackage{silence}
\usepackage[colorlinks=true,urlcolor=blue,citecolor=blue,linkcolor=blue]{hyperref}

\mathtoolsset{showonlyrefs}

\usepackage{comment}

\newtheorem{theorem}{Theorem}[section]
\newtheorem{remark}{Remark}[section]

\begin{document}

\title{Binary Optimization with Complex Constraints\\
via Quantum Approximate Multi-Objective Optimization}
\author{Andres Ruiz}
\affiliation{IBM Research}

\author{Soumyadip Ghosh}
\affiliation{IBM Research}

\author{Stefan Woerner}
\affiliation{IBM Research}

\begin{abstract}
We show that a class of binary optimization problems with complex non-quadratic objectives or constraints can be reformulated as multi-objective quadratic unconstrained binary optimization problems. When the objective and constraints depend on a small number of quadratic features and are monotone with respect to their preferred directions, at least one globally optimal solution lies in the Pareto set of the associated MO-QUBO. This enables the constraints to be evaluated classically on Pareto-optimal candidates rather than encoded as penalties. We demonstrate the approach for binary portfolio optimization under a Conditional Value-at-Risk constraint. Using Quantum Approximate Multi-Objective Optimization on an illustrative 100-asset instance, we approximate the mean-variance Pareto front using an IBM Quantum computer and derive mean-CVaR fronts through classical post-processing. The hardware results recover the overall structure of the classical front and yield near-optimal feasible portfolios for different risk bounds.
\end{abstract}

\date{\today}

\maketitle

\section{Introduction}

We show how a class of binary optimization problems with complex non-quadratic objectives or constraints can be reformulated as multi-objective quadratic unconstrained binary optimization (MO-QUBO) problems. 
Such objectives and constraints can be difficult to incorporate directly into a single-objective QUBO, as their representation through penalties may require auxiliary variables, additional approximations, or large penalty coefficients \cite{Lucas2014,GloverKochenbergerDu2019}. 
Here, we consider problems whose objective and constraints depend on the binary decision variables through a small number of quadratic feature functions, while the remaining dependence on these features may be non-quadratic. 

If the objective and constraints are monotone with respect to the preferred feature directions, then at least one globally optimal solution of the original constrained problem lies in the Pareto set of the associated MO-QUBO. 
The original problem can therefore be solved by computing the Pareto set, evaluating the non-quadratic objective and constraints classically on its elements, and selecting the best feasible solution. This avoids encoding the complex parts of the problem directly into the quantum optimization model. The framework can also be extended to constrained problems with multiple original objectives and to higher-order feature functions. 

The reduction follows from standard Pareto dominance arguments in multi-objective optimization: if the objective and feasibility are monotone with respect to the induced dominance order, dominated solutions are not needed to attain the global optimum \cite{Miettinen1999,Ehrgott2005}. Its relevance here is algorithmic, as it replaces a complex constrained optimization problem with the task of approximating the Pareto set of a collection of simpler objective functions. This task can be addressed, for example, using the Quantum Approximate Multi-Objective Optimization algorithm (QAMOO), which employs QAOA-style circuits to approximate Pareto fronts of multi-objective combinatorial optimization problems \cite{FarhiGoldstoneGutmann2014,Kotil2025QAMOO}. If the complete Pareto set is available, the reduction preserves global optimality. If only an approximation is available, the resulting procedure is heuristic and does not, in general, guarantee recovery of a globally optimal solution. For broader perspecctives on quantum optimization and its systematic benchmarking, see Refs.~\cite{Abbas2024QuantumOptimization, Koch2026QOBLIB}

We demonstrate the approach for binary portfolio optimization under a Conditional Value-at-Risk (CVaR) constraint. In this setting, the non-quadratic constraint depends only on two quadratic features: expected return and portfolio variance. The corresponding mean-variance Pareto set can therefore be approximated using a quantum algorithm, while the CVaR constraint is evaluated through classical post-processing. Under the same Gaussian return assumption, the approach applies analogously to Value-at-Risk (VaR) constraints \cite{RockafellarUryasev2000,McNeilFreyEmbrechts2015}. The approach can also be extended to more complex settings, e.g., including cardinality or budget constraints \cite{Lucas2014, vazquez2026efficientfourierbasedlinearcombination}.

\section{Problem Formulation}
\label{sec:problem_formulation}

We consider a class of constrained binary optimization problems whose objective and constraints are expressed through a collection of quadratic features.
Let $X\subseteq\{0,1\}^n$ denote the feasible binary search space. 
The set $X$ may be the full hypercube $\{0,1\}^n$, but it may also encode constraints that are easily incorporated into a QUBO formulation, such as strict cardinality constraints \cite{Lucas2014, GloverKochenbergerDu2019}. 
Here, our focus is on more general constraints that do not naturally admit a quadratic penalty representation.

Let \(h_i:X\rightarrow\mathbb{R},\; i=1,\ldots,m,\) denote feature functions. 
We assume each $h_i$ is at most quadratic in the binary decision variables, so that it can be represented as a QUBO objective, although extensions to higher-order terms are possible.
Then, the optimization problem of interest is
\begin{align}
\max_{x \in X} & \quad F(h_1(x),\ldots,h_m(x))
\label{eq:original_problem}\\
\text{s.t.}\;& \quad
G_j(h_1(x),\ldots,h_m(x)) \geq 0,
\quad j=1,\ldots,J.
\end{align}
We write these compactly as $F(h(x))$ and $G_j(h(x))$.

The key idea is to reinterpret \eqref{eq:original_problem} as a multi-objective optimization problem. 
For each feature $h_i$, let $s_i\in\{-1,+1\}$ encode its preferred direction, where the signs are chosen such that the objective $F$ and all constraint functions $G_j$ are monotone nondecreasing with respect to the induced dominance order. 
We then define \[ Z_i(x)=s_i h_i(x), \qquad i=1,\ldots,m. \] 
This yields the multi-objective QUBO (MO-QUBO)
\begin{align}
\max_{x\in X} \quad \{Z_1(x),\ldots,Z_m(x)\}.
\label{eq:moqubo}
\end{align}
If choosing consistent preferred directions for all features is not possible, the assumptions underlying the reduction are not satisfied and the theory developed below does not apply. The resulting MO-QUBO formulation can nevertheless still be used heuristically.

The reduction relies on the observation that, for monotone objectives and constraints, any dominated solution in the transformed feature space can be replaced by a dominating solution without decreasing the objective or worsening feasibility. 
We say that $x'\in X$ dominates $x\in X$ if $Z_i(x')\ge Z_i(x)$ for all $i$, with strict inequality for at least one component. 
The Pareto set $\mathcal P\subseteq X$ consists of all non-dominated solutions. 
The next theorem makes this reduction precise.

\begin{theorem}
\label{thm:pareto_reduction}
Let $X$ be finite, and let $\mathcal P$ denote the Pareto set of the
MO-QUBO \eqref{eq:moqubo}.
Assume that $F$ and all $G_j$ are monotone nondecreasing with respect to the dominance order induced by $Z = (Z_1, \ldots, Z_m)$, i.e., whenever $Z(x')\ge Z(x)$ componentwise, then
\[
F(h(x'))\ge F(h(x))
\]
and
\[
G_j(h(x'))\ge G_j(h(x)),
\qquad j=1,\ldots,J.
\]
If \eqref{eq:original_problem} is feasible, then at least one globally optimal solution belongs to $\mathcal P$. 
Equivalently,
\[
\max_{\substack{x\in X\\ G_j(h(x))\ge 0,\ \forall j}} F(h(x))
\quad=
\max_{\substack{x\in\mathcal P\\ G_j(h(x))\ge 0,\ \forall j}} F(h(x)).
\]
\end{theorem}

\begin{proof}
The ``$\geq$'' direction follows immediately from $\mathcal P\subseteq X$.
For the reverse inequality, let $x^\star$ be an optimal feasible solution
of \eqref{eq:original_problem}. If $x^\star\in\mathcal P$, then we trivially have the assertion. 
Otherwise, there exists $x'\in X$ that dominates
$x^\star$. By monotonicity,
\[
G_j(h(x'))\ge G_j(h(x^\star))\ge 0,
\qquad \forall j,
\]
so $x'$ is feasible. Similarly,
\[
F(h(x'))\ge F(h(x^\star)).
\]
Since $x^\star$ is globally optimal, $x'$ is also globally optimal. If
$x'$ is not Pareto optimal, the argument can be repeated. Because $X$ is
finite, this process terminates at a Pareto-optimal feasible solution with
the same objective value as the global optimum. Hence, at least one
globally optimal feasible solution is attained within $\mathcal P$.
\end{proof}

\begin{remark}
Thm.~\ref{thm:pareto_reduction} is a standard consequence of Pareto dominance theory \cite{Miettinen1999,Ehrgott2005}. 
Its relevance here is algorithmic: whenever the objective and constraints are monotone with respect to the preferred feature directions encoded by $Z$, the original constrained binary optimization problem can be solved by searching the Pareto set of the corresponding MO-QUBO.
\end{remark}

\begin{remark}
Thm.~\ref{thm:pareto_reduction} guarantees that at least one globally optimal solution of \eqref{eq:original_problem} lies in the Pareto set, but not necessarily all globally optimal solutions. Non-Pareto-optimal global optima can occur when the optimum of the original single-objective problem is non-unique. In that case, a dominating Pareto-optimal representative is preferable: by monotonicity, it cannot decrease the objective or worsen any constraint value. Hence, the Pareto-based formulation naturally selects feature-wise efficient optima, which may have larger feasibility margins and therefore be more robust with respect to the modeled constraints.
\end{remark}

\section{Portfolio Optimization with CVaR Constraint}
\label{sec:cvar_portfolio}

As a concrete example, we consider binary portfolio optimization under a CVaR constraint under the assumption that returns follow a joint Gaussian distribution. 
The same construction can also be applied to a Value-at-Risk (VaR) constraint under the same Gaussian return assumption. 

Let $R\in\mathbb{R}^n$ denote the random vector of asset returns, and assume \( R\sim\mathcal N(\mu,\Sigma), \) where $\mu \in \mathbb{R}^n$ is the vector of expected returns and $\Sigma \in \mathbb{R}^{n\times n}$ is the covariance matrix. 
Binary decision variables $x\in\{0,1\}^n$ indicate asset selection, where $x_i=1$ means that asset $i$ is included in the portfolio. 
The corresponding portfolio return and loss are
\[ R^T x \qquad\text{and}\qquad L(x)=-R^T x, \] 
respectively. Additional admissibility constraints, such as cardinality or budget constraints, can be encoded in $X\subseteq\{0,1\}^n$. 

The considered optimization problem is
\begin{align} 
\max_{x\in X} &\quad \mu^T x \\ 
\text{s.t.} &\quad \operatorname{CVaR}_{\alpha}(L(x)) \le c. 
\label{eq:cvar_portfolio}
\end{align} 
We use the upper-tail CVaR convention for losses \cite{RockafellarUryasev2000, AcerbiTasche2002, McNeilFreyEmbrechts2015}. For a loss random variable $L$ and confidence level $\alpha\in(0,1)$, let 
\[
\operatorname{VaR}_{\alpha}(L) = \inf\{\ell\in\mathbb{R}:\mathbb{P}(L\le \ell)\ge \alpha\}. 
\]
For continuous loss distributions, the corresponding CVaR is the tail expectation
\[ \operatorname{CVaR}_{\alpha}(L) = \mathbb{E}\left[L\,\middle|\,L\ge \operatorname{VaR}_{\alpha}(L)\right]. 
\] 
Since $R$ is Gaussian, the loss satisfies 
\[ 
L(x) \sim \mathcal N\!\left(-\mu^T x,\; x^T\Sigma x\right). 
\] 
Hence, for confidence level $\alpha\in(0,1)$, the CVaR is 
\begin{align}  
\operatorname{CVaR}_{\alpha}(L(x)) &= -\mu^T x + \kappa_\alpha \sqrt{x^T\Sigma x}, \label{eq:gaussian_cvar}\\ \kappa_\alpha &= \frac{\phi(\Phi^{-1}(\alpha))}{1-\alpha}, 
\end{align} 
where $\phi$ and $\Phi$ denote the standard normal probability density and cumulative distribution functions, respectively. The VaR can be expressed analogously with \(\kappa_{\alpha}=\Phi^{-1}(\alpha)\)
\cite{RockafellarUryasev2000, McNeilFreyEmbrechts2015}.

Substituting \eqref{eq:gaussian_cvar} into \eqref{eq:cvar_portfolio} yields
\begin{align}
\max_{x\in X} &\quad \mu^T x \\
\text{s.t.} &\quad
-\mu^T x + \kappa_\alpha \sqrt{x^T\Sigma x} \le c.
\label{eq:cvar_portfolio_explicit}
\end{align}
The constraint in \eqref{eq:cvar_portfolio_explicit} is a non-quadratic inequality constraint and therefore does not naturally admit a simple penalty formulation.
However, it depends only on two at most quadratic quantities:
\begin{align}
h_1(x) &= \mu^T x,\\
h_2(x) &= x^T\Sigma x.
\end{align}
The first feature measures expected return and should be maximized, while the second measures variance and should be minimized, as both changes improve the objective and increase feasibility with respect to the CVaR constraint.
Following Sec.~\ref{sec:problem_formulation}, we define
\begin{align}
Z_1(x) &= h_1(x)=\mu^T x,\\
Z_2(x) &= -h_2(x)=-x^T\Sigma x,
\end{align}
which leads to the MO-QUBO
\begin{equation}
\max_{x\in X}
\left\{
\mu^T x,\,
-x^T\Sigma x
\right\}.
\label{eq:cvar_moqubo}
\end{equation}
This is closely related to the classical mean-variance frontier of Markowitz portfolio theory \cite{Markowitz1952}.

Expressing the objective and constraint in terms of
$(z_1,z_2)=(Z_1(x),Z_2(x))$ gives
\[
F(z_1,z_2)=z_1
\]
and
\[
G(z_1,z_2)=c+z_1-\kappa_\alpha\sqrt{-z_2}.
\]
Since $\Sigma$ is positive semidefinite, $z_2\le 0$. Moreover, $G$ is
nondecreasing in both arguments: increasing $z_1$ increases expected
return, while increasing $z_2$ corresponds to decreasing variance. 
The objective $F$ is also nondecreasing with respect to the induced dominance order.

Therefore, the assumptions of Thm.~\ref{thm:pareto_reduction} are
satisfied. 
Consequently, at least one globally optimal solution of \eqref{eq:cvar_portfolio_explicit} lies on the Pareto set of the MO-QUBO \eqref{eq:cvar_moqubo}. 
If the complete Pareto set is available, solving the original problem reduces to evaluating the CVaR constraint on the Pareto solutions and selecting the feasible portfolio with maximum expected return. 
If only an approximation of the Pareto set is available, the same filtering procedure becomes heuristic. It returns the best feasible candidate among the generated solutions, but global optimality is not guaranteed in general.

\section{QAMOO for CVaR-Constrained Portfolio Optimization}
\label{sec:qamoo_cvar}

The previous section showed that the CVaR-constrained binary portfolio problem can be reduced to searching the Pareto set of the bi-objective QUBO
\[
\max_{x\in X} \left\{ \mu^T x,\, -x^T\Sigma x \right\}.
\]
This motivates the use of the Quantum Approximate Multi-Objective Optimization algorithm (QAMOO), which uses QAOA-style circuits to approximate Pareto fronts of multi-objective combinatorial optimization problems \cite{FarhiGoldstoneGutmann2014, Kotil2025QAMOO}.

In this setting, QAMOO generates Pareto-efficient portfolio candidates, while the non-quadratic CVaR constraint is evaluated classically. 
The overall procedure is:
\begin{enumerate}
    \item Construct the MO-QUBO objectives
    \[
    Z_1(x)=\mu^T x,
    \qquad
    Z_2(x)=-x^T\Sigma x,
    \]
    on the admissible portfolio set $X$.

    \item Apply QAMOO to obtain candidate portfolios approximating the
    Pareto set of \eqref{eq:cvar_moqubo}.

    \item For each generated candidate, evaluate the CVaR constraint
    \[
    -\mu^T x+\kappa_\alpha\sqrt{x^T\Sigma x}\le c.
    \]

    \item Among the feasible candidates, select the portfolio with the largest expected return $\mu^T x$.
\end{enumerate}

If QAMOO returns the complete Pareto set, the filtering step recovers a globally optimal solution of the original CVaR-constrained problem by Thm.~\ref{thm:pareto_reduction}. 
With an approximate Pareto set, the same procedure becomes a heuristic: global optimality is not guaranteed, but
the approach focuses the classical feasibility check on portfolios that are efficient in the mean-risk feature space.

\section{Numerical Results}
\label{sec:numerical_results}

We demonstrate the proposed approach on an illustrative portfolio instance with $n=100$ assets. 
We construct a sparse, hardware-compatible positive semidefinite covariance matrix $\Sigma\in\mathbb{R}^{n\times n}$ whose quadratic interactions can be implemented on a linear qubit topology using $k=4$ SWAP layers, together with an expected-return vector $\mu\in\mathbb{R}^n$. 
Further details on the construction of $\mu$ and $\Sigma$ are provided in Appendix~\ref{sec:instance_construction}. 
These quantities define the expected-return and variance objectives, $\mu^T x$ and $x^T\Sigma x$, respectively, for $x\in\{0,1\}^n$.

\begin{figure*}[ht]
    \centering
    \includegraphics[width=\linewidth]
    {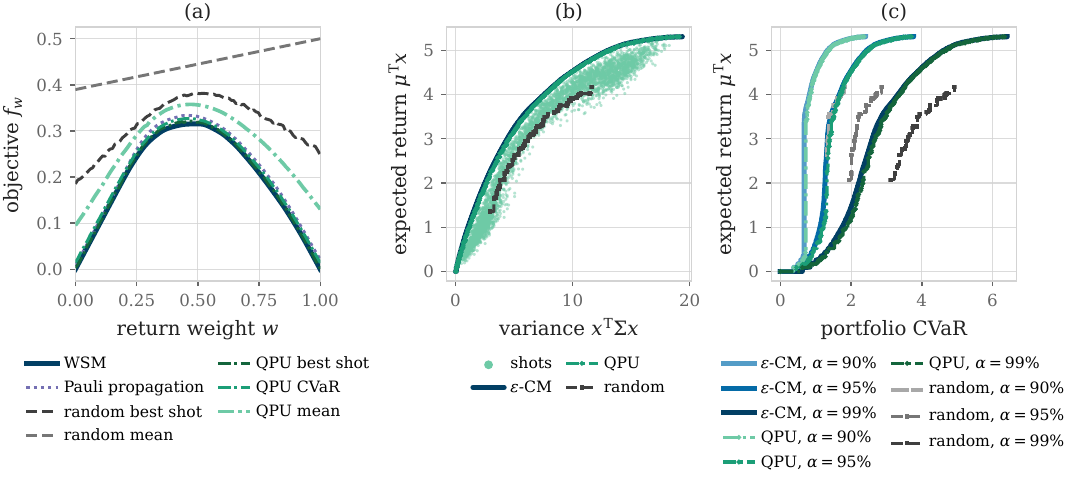}
    \caption{Results for the 100-asset portfolio instance. 
    (a) Scalarized objective values for 100 equidistant weights $w\in[0,1]$, showing the WSM reference, Pauli-propagation values, QPU mean, QPU-CVaR, QPU best shot, and random baselines. 
    The QPU-CVaR curve averages the objective over the best $0.326\%$ of the $100{,}000$ shots. 
    (b) Mean-variance values of all measured samples, together with the QPU, $\epsilon$-CM, and random non-dominated fronts. 
    (c) Mean-CVaR fronts of the $\epsilon$-CM, QPU, and random solutions for confidence levels $\alpha=90\%$, $95\%$, and $99\%$. All $\epsilon$-CM curves and fronts are derived from the same 1,000 $\epsilon$-CM solutions obtained by CPLEX; weighted-sum results are reported separately in Tab.~\ref{tab:results} and Appendix~\ref{sec:cplex_results}.}
    \label{fig:results}
\end{figure*}

Because expected return and variance may have substantially different scales, we normalize the two objectives using objective values of the anchor solutions. 
We first consider the minimization objectives
\begin{align}
    f_1(x) &= -\mu^T x, \\
    f_2(x) &= x^T\Sigma x,
\end{align}
and minimize each objective individually. 
For the instance considered here, all expected returns are positive and $\Sigma$ is positive semidefinite.
Therefore, the anchor solutions can be chosen as
$x_1^\star=\mathbf{1}$ and $x_2^\star=\mathbf{0}$, representing a maximum-return and a minimum-variance portfolio, respectively. 
We define the normalized objectives as
\begin{align}
    \widetilde f_1(x)
    &=
    \frac{f_1(x)-f_1(x_1^\star)}
         {f_1(x_2^\star)-f_1(x_1^\star)}, \\
    \widetilde f_2(x)
    &=
    \frac{f_2(x)-f_2(x_2^\star)}
         {f_2(x_1^\star)-f_2(x_2^\star)}.
\end{align}
Each normalized objective is zero at its own anchor and one at the opposite anchor.

We evaluate the quality of the approximated Pareto fronts using the hypervolume (HV) indicator \cite{riquelme_2015_moo_metrics}. In the present bi-objective setting, the HV is the area of the objective space that is dominated by the non-dominated solutions and bounded by a fixed reference point. Before computing the HV, we express both objectives in minimization form and use the anchor normalization introduced above. We choose the reference point $(1.05,1.05)$, which is slightly worse than the two opposite anchor values $(1,1)$. All solution sets are evaluated using the same normalization and reference point. A larger HV indicates that the approximation is closer to the Pareto front, covers a broader range of tradeoffs, or both.

To obtain a classical reference, we apply the $\epsilon$-constraint method ($\epsilon$-CM)~\cite{Ehrgott2005} using IBM ILOG CPLEX~\cite{cplex}. 
We solve
\begin{align}
    \min_{x\in\{0,1\}^n} &\quad \widetilde f_2(x) \\
    \text{s.t.} &\quad \widetilde f_1(x)\leq\epsilon
    \label{eq:eps-cm}
\end{align}
for $1{,}000$ equidistant values of $\epsilon\in[0,1]$. 
The $\epsilon$-CM can recover both supported and unsupported Pareto-optimal solutions, where supported solutions lie on the lower convex envelope of the attainable objective set and unsupported solutions correspond to non-convex regions of the Pareto front \cite{Ehrgott2005}.
The results are shown in Fig.~\ref{fig:results} and Tab.~\ref{tab:results}.
CPLEX solves all of the resulting problems to proven optimality within the given time.
Further details are provided in Appendix~\ref{sec:cplex_results}.

\begin{table}
    \centering
    \resizebox{\columnwidth}{!}{%
        \begin{tabular}{lcccc} 
            Mean-Variance & $\epsilon$-CM & QAMOO & WSM & Random \\ 
            \hline 
            $\#$Points & 901 & 457 & 64 & 69 \\ 
            HV & 0.8482 & 0.8411 & 0.8400 & 0.6617 \\ 
            Rel.~HV & 1.0000 & 0.9917 & 0.9904 & 0.7802 \\ 
            &&&&\\ 
            Mean-CVaR (90\%) & $\epsilon$-CM & QAMOO & WSM & Random \\ 
            \hline 
            $\#$Points & 227 & 127 & 45 & 15 \\ 
            HV & 0.7751 & 0.7560 & 0.7723 & 0.4038 \\ 
            Rel.~HV & 1.0000 & 0.9753 & 0.9964 & 0.5209 \\
            &&&&\\ 
            Mean-CVaR (95\%) & $\epsilon$-CM & QAMOO & WSM & Random \\ 
            \hline 
            $\#$Points & 448 & 185 & 63 & 21 \\ 
            HV & 0.7172 & 0.7015 & 0.7124 & 0.4143 \\ 
            Rel.~HV & 1.0000 & 0.9781 & 0.9932 & 0.5776 \\ 
            &&&&\\ 
            Mean-CVaR (99\%) & $\epsilon$-CM & QAMOO & WSM & Random \\ 
            \hline 
            $\#$Points & 635 & 245 & 63 & 30 \\ 
            HV & 0.6891 & 0.6755 & 0.6824 & 0.4271 \\ 
            Rel.~HV & 1.0000 & 0.9803 & 0.9902 & 0.6198
        \end{tabular}
    }
    \caption{Pareto-front quality for the mean-variance and mean-CVaR formulations. Results are shown for $\epsilon$-CM with 1,000 discretization points, QAMOO at depth $p=3$, WSM using the same 100 scalarization weights, and uniform random sampling with $10^7$ samples, matching the total QAMOO shot count. \emph{$\#$Points} denotes the number of non-dominated portfolios found by each approach. Since the mean-CVaR front is a subset of the mean-variance front, the counts differ between blocks. \emph{HV} denotes the normalized hypervolume computed using anchor-based normalization and reference point $(1.05,1.05)$; the ranges are $[0,5.3097]$ for return, $[0,19.3635]$ for variance, and $[0,2.4129]$, $[0,3.7671]$, and $[0,6.4183]$ for CVaR at $\alpha=90\%$, $95\%$, and $99\%$, respectively. \emph{Rel.~HV} is normalized by the corresponding $\epsilon$-CM value and therefore equals one for $\epsilon$-CM by construction. The Random column reports a representative run selected as the lower median of ten independent runs ranked by mean-variance HV; across these runs, the mean-variance HV ranged from $0.655940$ to $0.668749$.}
    \label{tab:results}
\end{table}

We further consider the discretized weighted-sum method (WSM) with 100 equidistant scalarization weights $w\in[0,1]$. 
For each weight, we minimize
\begin{equation}
    f_w(x)
    =
    w\,\widetilde f_1(x)
    +(1-w)\,\widetilde f_2(x).
    \label{eq:wsm}
\end{equation}
We solve the resulting QUBO for each value of $w$ using CPLEX~\cite{cplex}, which proves optimality for all scalarizations. 
Further details are provided in Appendix~\ref{sec:cplex_results}. 
In contrast to the $\epsilon$-CM, the WSM recovers only supported Pareto-optimal solutions.
Again, the results are shown in Fig.~\ref{fig:results} and Tab.~\ref{tab:results}.

Finally, we run QAMOO~\cite{Kotil2025QAMOO} by using QAOA as a heuristic solver within the WSM for the same 100 scalarization weights as before. 
For each scalarization, we apply QAOA~\cite{FarhiGoldstoneGutmann2014} with the linear-ramp (LR) parameterization \cite{montanez2025lrrqaoa} at depth $p=3$. 
The corresponding energy is evaluated using Pauli propagation~\cite{Rudolph2026PauliPropagation,Angrisani2025NoiselessPauliPropagation} and the two LR slopes are optimized using a grid search followed by a local refinement: eight well-separated low-energy grid cells seed eight shrinking-step local searches, and the lowest of the eight is kept. Further details on the LR-QAOA schedule and parameter selection are provided in Appendix~\ref{sec:lr_qaoa}.

Each circuit is executed with $100{,}000$ shots and a repetition delay of $100\,\mu\mathrm{s}$ on \emph{ibm\_boston} through the IBM Quantum
Platform~\cite{ibm_quantum_platform}. 
With 100 scalarization weights, this corresponds to
\(
    100\times100{,}000=10{,}000{,}000
\)
shots. 
After transpilation, each circuit acts on 100 qubits and contains 2,376 CZ gates at a two-qubit-gate depth of 48. 
As a heuristic estimate, we multiply the reported CZ-gate fidelities, single-qubit-gate fidelities, and measurement fidelities associated with the transpiled circuit \cite{Barron_2024}. Under an independent stochastic-error interpretation, the resulting product approximates the probability that an execution of the circuit suffers no error. For $p=3$ this estimate is $0.326\%$, corresponding to approximately 326 of the $100{,}000$ shots per weight. The aggregated QPU execution time over all circuits was 1,194~s,
excluding queueing time.

Fig.~\ref{fig:results}(a) shows that the measured expectation values follow the trend of the CPLEX and Pauli-propagation references, although hardware noise increases the expected scalarized objective. 
We additionally compute the lower-tail CVaR of the sampled scalarized objective distribution, setting its tail fraction to the estimated circuit fidelity of $0.326\%$. This shot tail fraction is unrelated to the portfolio-CVaR confidence levels $\alpha=90\%,95\%,99\%$ of Fig.~\ref{fig:results}(c), which refer to the portfolio's loss distribution and not to shots. As shown by Barron et al.~\cite{Barron_2024}, under the corresponding noise-model assumptions, this CVaR estimated using noisy samples provides a provable lower bound on the noise-free expectation value. The observed CVaR values lie below the Pauli-propagation values for the 79 weights $w\geq0.212$; for the 21 smallest weights, $w\leq0.202$, where the Pauli-propagation value is closest to the CPLEX optimum, they exceed it by at most $0.0072$. Moreover, the best sampled solutions closely track the CPLEX reference solutions across the scalarization weights.

For every sampled portfolio $x$, we evaluate its expected return $\mu^T x$ and variance $x^T\Sigma x$. 
The non-dominated sampled portfolios form an approximation of the mean-variance Pareto front. 
Fig.~\ref{fig:results}(b) compares this front with the $\epsilon$-CM reference front and shows that the quantum samples closely recover its overall structure. 
As reported in Tab.~\ref{tab:results}, QAMOO achieves $99.17\%$ of the mean-variance HV of the $\epsilon$-CM reference, compared with $99.04\%$ for CPLEX-based WSM. 
Although QAMOO uses the same weighted scalarizations as WSM, each circuit generates many portfolios rather than a single optimum, and can recover non-dominated portfolios beyond the supported solutions returned by WSM.

Finally, for each generated portfolio, we evaluate 
\begin{equation} 
\operatorname{CVaR}_{\alpha}(L(x)) = -\mu^T x +\kappa_\alpha\sqrt{x^T\Sigma x} 
\end{equation} 
for confidence levels $\alpha = 90\%, 95\%, 99\%$, corresponding to upper-tail probabilities of $10\%$, $5\%$, and $1\%$, respectively.
The corresponding values are $\kappa_{0.90}=1.7550$, $\kappa_{0.95}=2.0627$, and $\kappa_{0.99}=2.6652$.
Since $\alpha$ enters only through $\kappa_\alpha$, all three CVaR values are computed classically from the same portfolios. 
For each confidence level, we transform the portfolios from the mean-variance representation to the mean-CVaR representation.
To compute the HV of each mean-CVaR front, we again express both objectives in minimization form, apply anchor-based normalization, and use $(1.05, 1.05)$ as reference point.
The resulting fronts are shown in Fig.~\ref{fig:results}(c), with the corresponding HV values reported in Tab.~\ref{tab:results}. 
QAMOO reaches between $97.53\%$ and $98.03\%$ of the corresponding $\epsilon$-CM reference HVs, while WSM reaches between $99.02\%$ and $99.64\%$.
Uniform random sampling yields substantially lower HV than the optimization-based approaches. 

To assess robustness, we performed ten additional repetitions of the $p=3$, $k=4$ hardware experiment directly after the main run. The median curves and variability across repetitions are reported in Appendix~\ref{sec:run_stability}. The corresponding median results are nearly indistinguishable from those of the main hardware run reported in Fig.~\ref{fig:results} and Tab.~\ref{tab:results}.


\begin{figure}[ht]
\centering
\includegraphics[width=\columnwidth]{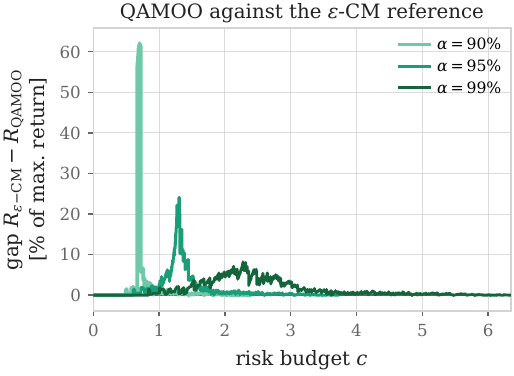}
\caption{Expected-return gap between the QAMOO and $\epsilon$-CM mean-CVaR fronts
as a function of the CVaR budget $c$, for $p=3$ and $k=4$ on \emph{ibm\_boston}.
For each budget $c$, the return $R(c)$ is the maximum expected return among
portfolios satisfying $\operatorname{CVaR}_{\alpha}(L(x))\le c$, evaluated in the
Gaussian form $\operatorname{CVaR}_{\alpha}(L(x))=-\mu^{\mathrm{T}}x
+\kappa_{\alpha}\sqrt{x^{\mathrm{T}}\Sigma x}$ with
$\kappa_{\alpha}=\varphi(\Phi^{-1}(1-\alpha))/\alpha$. The plotted quantity is
the signed difference \((R_{\epsilon\text{-CM}}(c)-R_{\mathrm{QAMOO}}(c))/
\mu_{\max},\) where $\mu_{\max}=5.3097$ is the maximum attainable return,
reached by holding all 100 assets; positive values mean that the $\epsilon$-CM
reference attains the higher return at that budget. Results are shown for
$\alpha=90\%$, $95\%$, and $99\%$, over the budget ranges $c\le2.37$, $3.71$,
and $6.35$, respectively. Averaged over the respective budget ranges, the gap is
$1.950\%$, $1.592\%$, and $1.368\%$. The largest gaps lie where the
$\epsilon$-CM front rises steeply at an almost constant CVaR, reaching a maximum
gap of $62.10\%$ at $c=0.70$ for $\alpha=90\%$; gaps above $5\%$ occupy
$3.8\%$, $8.3\%$, and $6.8\%$ of the respective budget ranges. For
$2.1\%$, $1.7\%$, and $1.4\%$ of the budget range the sign reverses and
$R_{\mathrm{QAMOO}}(c)>R_{\epsilon\text{-CM}}(c)$, because the $\epsilon$-CM
reference is itself based on a finite discretization of 1,000 $\epsilon$ values
and does not represent the complete Pareto front; in that region the QAMOO
advantage never exceeds $0.16\%$ of $\mu_{\max}$, which is why the curve departs
only slightly below zero.}
\label{fig:budget_error}
\end{figure}

Given a CVaR upper bound $c$, the approximate solution of the original
constrained problem is the portfolio on the (approximated) mean-CVaR front with the largest expected return among all portfolios satisfying
\(
    \operatorname{CVaR}_{\alpha}(L(x))\leq c.
\)
Let $R(c)$ denote the corresponding return. Since only finitely many portfolios appear on the (approximated) mean-CVaR front, $R(c)$ is a step function that changes only when $c$ reaches the CVaR value of another front point. Different values of $c$ therefore select portfolios between the minimum-variance and maximum-return anchors. Figure~\ref{fig:budget_error} shows the normalized return gap between the $\epsilon$-CM and QAMOO solutions
\[ 
\frac{R_{\epsilon\text{-CM}}(c)-R_{\mathrm{QAMOO}}(c)}{\mu_{\max}}, 
\] 
where $\mu_{\max}=5.3097$ is the maximum attainable return. Averaged over the full range of $c$, the gap is $1.950\%$, $1.592\%$, and $1.368\%$ for $\alpha=90\%$, $95\%$, and $99\%$, respectively. The maximum gap of $62.10\%$ occurs for $\alpha=90\%$ near $c=0.70$, in a region where gaps above $5\%$ span $3.8\%$ of the budget range. 

Overall, QAMOO closely approximates the mean-variance and mean-CVaR reference fronts for this instance and achieves a higher mean-variance HV than the exact discretized WSM using the same scalarization weights, while WSM attains the higher mean-CVaR HVs. From the perspective of the original CVaR-constrained optimization problem, it yields portfolios whose expected returns are close to the $\epsilon$-CM reference over most of the range of CVaR budgets. QAMOO also generates a larger set of non-dominated portfolio candidates than WSM, which may provide additional flexibility when selection criteria beyond those included in the optimization model are considered, although these candidates are not necessarily non-dominated with respect to the combined reference set.

\section{Conclusion \& Outlook}

We introduced a framework for addressing binary optimization problems with complex non-quadratic objectives or constraints through a multi-objective QUBO reformulation. 
When the objective and constraints depend on a small number of quadratic features and are monotone with respect to their preferred directions, at least one globally optimal solution lies in the Pareto set of the associated MO-QUBO. 
Complex objectives and constraints can therefore be evaluated through classical post-processing rather than encoded directly as quadratic penalties.

We demonstrated the approach for binary portfolio optimization under a CVaR constraint using an IBM quantum computer. Applying QAMOO to an illustrative 100-asset instance, we approximated the mean-variance Pareto front and derived mean-CVaR fronts for different confidence levels. 
The quantum hardware results recovered the overall structure of the classical front and provided competitive feasible portfolio candidates across different risk bounds.

While the present framework focuses on quadratic features, it can be extended to higher-order feature functions. Future work should also focus on developing more efficient strategies for selecting QAOA parameters across scalarizations, rather than optimizing them independently for each weight.
Although demonstrated on a relatively simple single-objective problem, the framework is general and can help to enable quantum approaches to more complex and realistic constrained optimization problems.

\subsection*{Acknowledgments} 

Anthropic Claude Code and GPT-5.6 (via Microsoft Copilot) were used to assist with code development and editorial improvements during the preparation of this manuscript. All scientific ideas, analyses, and conclusions are those of the authors. All AI-assisted code and text were carefully reviewed, manually verified, and validated prior to inclusion.

\appendix

\section{Instance Construction}
\label{sec:instance_construction}

We generated the 100-asset illustrative instance used in Sec.~\ref{sec:numerical_results} by applying a hardware-compatible sparsification to a dense portfolio dataset. The dense dataset is randomly generated and does not use any market data. For $n=100$ assets, we use the constant-correlation model~\cite{elton1973}: each asset $i$ has its own annual volatility, and all pairs of assets share a single correlation $\rho$. Expected annual returns are drawn independently:
\begin{gather}
  \left(\Sigma_{\mathrm{dense}}\right)_{ij} =
  \begin{cases}
    \sigma_i^2, & i = j,\\
    \rho\, \sigma_i \sigma_j, & i \neq j,
  \end{cases}
  \; \sigma_i \sim \mathcal{U}(0.1,\, 0.3), \; \rho = 0.3,
  \\
  \mu_i \sim \mathcal{U}(0,\, 0.1).
\end{gather}
We denote the resulting expected-return vector and dense covariance matrix by $\mu\in\mathbb{R}^{100}$ and $\Sigma_{\mathrm{dense}}\in\mathbb{R}^{100\times100}$, respectively. The dense covariance matrix is symmetric and positive definite, with eigenvalues in $[0.0071,1.4715]$. 
All expected returns are positive and lie in $[0.0006,0.0995]$, with sum $5.3097$.

The interaction pattern is generated by $k=4$ layers of non-overlapping SWAP gates, alternating between layers starting on even- and odd-indexed qubits, applied to a line of 100 qubits \cite{Weidenfeller_2022}. 
This SWAP strategy determines which off-diagonal entries of $\Sigma_{\mathrm{dense}}$ are retained, which we capture in the interaction mask $M \in \{0, 1\}^{100 \times 100}$, where $1$ indicates that the corresponding quadratic interaction is possible and $0$ otherwise.
The corresponding circuit contains 198 SWAP gates, realizes 297 quadratic interactions (out of the 4,950 interactions in the original dense matrix), and is transpiled into 792 CZ gates per QAOA layer, i.e. the 2,376 CZ gates of the $p=3$ circuit quoted in Sec.~\ref{sec:numerical_results}.

We assign the assets, i.e., binary variables, to virtual qubits on the 100-qubit line using a permutation $\pi$ that
maximizes the retained squared covariance mass,
\begin{equation}
  \operatorname{score}(\pi)
  =
  \sum_{i \neq j}
  \left(\Sigma_{\mathrm{dense}}\right)_{\pi(i)\,\pi(j)}^2
  M_{ij}.
  \label{eq:qapscore}
\end{equation}
We solve this quadratic assignment problem using SciPy's \emph{Fast Approximate QAP} algorithm followed by \emph{2-opt refinement} \cite{virtanen2020scipy}. This retains $7.6832\%$ of the total squared off-diagonal covariance mass and $6.5442\%$ of the total absolute off-diagonal covariance mass.

Removing covariance terms can make the sparse matrix indefinite. We restore
positive semidefiniteness using
\begin{equation}
  \begin{aligned}
  \delta&=
  \max\!\left(
  0,-\lambda_{\min}(\Sigma_{\mathrm{sparse}})+10^{-8}
  \right),\\
  \Sigma&=\Sigma_{\mathrm{sparse}}+\delta I.
  \end{aligned}
  \label{eq:psdrepair}
\end{equation}
For the $k=4$ instance, $\delta=0.0606$, while the mean diagonal of
$\Sigma_{\mathrm{sparse}}$ is $0.0476$. 
This repair preserves all off-diagonal covariance terms while restoring positive semidefiniteness.
The resulting $\mu$ and $\Sigma$ are
those used in the MO-QUBO~\eqref{eq:cvar_moqubo} and in all numerical experiments reported in
Sec.~\ref{sec:numerical_results}.

\section{CPLEX Results}
\label{sec:cplex_results}

We used IBM~ILOG~CPLEX~22.2.0.0 with one thread, a relative MIP-gap target of
$10^{-4}$, and a three-hour time limit per model. Of the $1{,}000$
$\epsilon$-CM models in Eq.~\eqref{eq:eps-cm}, all 1,000 were proved optimal.
The median solve time was $0.0272$~s. The total solve time
was $32.7$~s. CPLEX
proved optimality for all 100 WSM scalarizations in Eq.~\eqref{eq:wsm}. The total solve time for these 100 models was 0.5232~s.

\section{LR-QAOA Parameter Training}
\label{sec:lr_qaoa}

For the QAOA implementation of Eq.~\eqref{eq:wsm}, we use the two-parameter
linear-ramp schedule of Ref.~\cite{montanez2025lrrqaoa},
\begin{equation}
  \begin{aligned}
  \beta_l&=-\left(1-\frac{l}{p}\right)\Delta\beta,\\
  \gamma_l&=\frac{l+1}{p}
  \frac{\Delta\gamma}{\kappa_{\mathrm{LR}}},
  \qquad l=0,\ldots,p-1,
  \end{aligned}
\end{equation}
with $p=3$, as in Sec.~\ref{sec:numerical_results}. The scaling factor is
\begin{equation}
  \kappa_{\mathrm{LR}}
  =
  \max\!\left(
  \max_i|a_i|,
  \max_{i<j}|J_{ij}|
  \right),
\end{equation}
where $a_i$ and $J_{ij}$ are the local fields and couplings of the Ising Hamiltonian corresponding to $f_w$.

For each of the 100 weights in Eq.~\eqref{eq:wsm}, we evaluate 121
geometrically spaced values of $\Delta\gamma$ in $[0.25,40]$ and 25 uniformly spaced values of $\Delta\beta$ in $[0,3\pi]$. We refine the grid solution with eight local searches. Their starting cells are selected in increasing energy order and must be separated from earlier starts by at least three grid steps on either axis. This corresponds to a factor of $1.135$ in $\Delta\gamma$ or $1.178$ rad in $\Delta\beta$ and helps sample different basins. From each start, the search evaluates the eight neighbors obtained by multiplying or dividing $\Delta\gamma$ by $1.0432$ and adding or subtracting $0.3927$ rad from $\Delta\beta$, including diagonal moves. It repeatedly moves to the lowest-energy neighbor until no improvement is found. Both steps are then halved: the $\Delta\gamma$ factor becomes $\sqrt{1.0432}$, while the $\Delta\beta$ step becomes $0.19635$ rad. After six halvings, the best of the eight searches is retained. Coordinates remain within the scanned window, and all energies use the same Pauli-propagation evaluation as the grid. For $p=3$ and $k=4$, refinement improved the energy for all 100 weights, with a median improvement of $2.3\times10^{-3}$. Relative to the improvement obtained by the grid search alone, this corresponds to a median of $1.2\%$ and a maximum of $1.7\%$. The energies are evaluated using Pauli propagation with maximum Pauli weight 10. Before propagation, the observable is normalized by its $L_1$-norm, and the absolute coefficient cutoff is set to $10^{-5}$ \cite{Rudolph2026PauliPropagation,Angrisani2025NoiselessPauliPropagation}. For $p=3$ and $k=4$, training the two LR slopes for all 100 weights required $3{,}280$ CPU-hours on a cluster. The trained slopes and associated quantities across the 100 weights are shown in Fig.~\ref{fig:slopes}.

\begin{figure*}[ht]
    \centering
    \includegraphics[width=\linewidth]
    {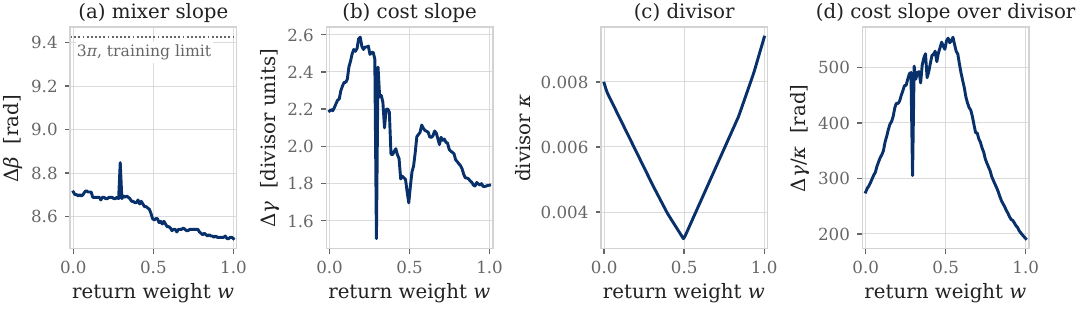}
    \caption{The trained ramp quantities at $p=3$, $k=4$, against the return weight $w$, over the 100 weights submitted to hardware. (a) the mixer slope $\Delta\beta$; the dotted line marks $3\pi$, the upper limit of the training range. (b) the cost slope $\Delta\gamma$, in units of the divisor. (c) the divisor $\kappa_{\mathrm{LR}}$, the largest absolute Ising coefficient of the problem at that weight; $\kappa_{\mathrm{LR}}$ depends on the problem and the weight only, not on the circuit depth. (d) the ratio $\Delta\gamma/\kappa_{\mathrm{LR}}$, in radians. Over the 100 weights, $\Delta\beta$ runs from 8.498 to 8.848 rad, a factor 1.04, with median 8.590 rad; $\Delta\gamma$ runs from 1.504 to 2.587 divisor units, a factor 1.72; $\kappa_{\mathrm{LR}}$ runs from 0.003184 to 0.009369, a factor 2.94, smallest at $w=0.495$ and largest at $w=1.000$; and $\Delta\gamma/\kappa_{\mathrm{LR}}$ runs from 191.2 to 554.3 rad, a factor 2.90. Markers are drawn every 13 weights.}
    \label{fig:slopes}
\end{figure*}

\section{Robustness of Results} 
\label{sec:run_stability}
To assess robustness, we repeated the $p=3$, $k=4$ hardware experiment ten times. 
Figure~\ref{fig:run_stability} shows the median curves and the variability across repetitions. 
These runs were executed on \emph{ibm\_boston} directly after the main run, with the same estimated circuit fidelity of $0.326\%$ reported in Sec.~\ref{sec:numerical_results}. 
The results exhibit the same qualitative behavior and remain very close to those reported in the main text. 

\begin{figure*}[ht]
    \centering 
    \includegraphics[width=\linewidth]{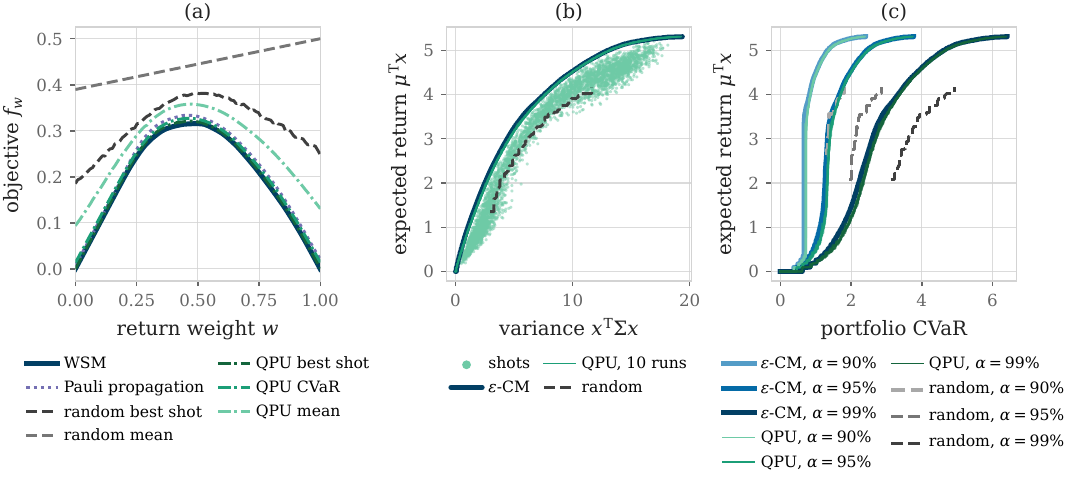} 
    \caption{Robustness analysis at $p=3$ and $k=4$ on \texttt{ibm\_boston}. The 100 circuits used in Sec.~\ref{sec:numerical_results} were executed ten times, with $100{,}000$ shots per weight in each repetition. QPU curves show the median across repetitions; the $\epsilon$-CM reference and random baseline are identical for all runs. The number of non-dominated portfolios ranged from 447 to 502 for the mean-variance front, and from 119 to 142, 170 to 194, and 229 to 268 for the mean-CVaR fronts at $\alpha=90\%$, $95\%$, and $99\%$, respectively. The QPU-CVaR tail fraction was set to the estimated circuit fidelity, $0.326\%$, as for the run reported in the main text. Overall, the repeated runs closely reproduce the results of Sec.~\ref{sec:numerical_results}.} 
    \label{fig:run_stability}
\end{figure*}

\bibliographystyle{arxiv_no_month}
\bibliography{references}

@book{Miettinen1999,
  author    = {Miettinen, Kaisa},
  title     = {Nonlinear Multiobjective Optimization},
  publisher = {Kluwer Academic Publishers},
  address   = {Boston},
  year      = {1999},
  url       = {https://users.jyu.fi/~miettine/book/}
}

@book{Ehrgott2005,
  author    = {Ehrgott, Matthias},
  title     = {Multicriteria Optimization},
  edition   = {2},
  publisher = {Springer},
  year      = {2005},
  doi       = {10.1007/3-540-27659-9},
  url       = {https://link.springer.com/book/10.1007/3-540-27659-9}
}

@article{Markowitz1952,
  author  = {Markowitz, Harry},
  title   = {Portfolio Selection},
  journal = {The Journal of Finance},
  volume  = {7},
  number  = {1},
  pages   = {77--91},
  year    = {1952},
  doi     = {10.1111/j.1540-6261.1952.tb01525.x},
  url     = {https://afajof.org/issue/volume-7-issue-1/}
}

@article{RockafellarUryasev2000,
  author  = {Rockafellar, R. Tyrrell and Uryasev, Stanislav},
  title   = {Optimization of Conditional Value-at-Risk},
  journal = {Journal of Risk},
  year    = {2000},
  volume  = {2},
  number  = {3},
  pages   = {21--41},
  doi     = {10.21314/JOR.2000.038},
  url     = {https://doi.org/10.21314/JOR.2000.038}
}

@article{Kotil2025QAMOO,
  author  = {Kotil, Ayse and Pelofske, Elijah and Riedmuller, Stephanie and Egger, Daniel J. and Eidenbenz, Stephan and Koch, Thorsten and Woerner, Stefan},
  title   = {Quantum Approximate Multi-Objective Optimization},
  journal = {Nature Computational Science},
  year    = {2025},
  doi     = {10.1038/s43588-025-00873-y},
  url     = {https://www.nature.com/articles/s43588-025-00873-y},
  eprint  = {2503.22797},
  archivePrefix = {arXiv},
  primaryClass  = {quant-ph}
}

@article{Lucas2014,
  author  = {Lucas, Andrew},
  title   = {Ising Formulations of Many {NP} Problems},
  journal = {Frontiers in Physics},
  year    = {2014},
  volume  = {2},
  pages   = {5},
  doi     = {10.3389/fphy.2014.00005},
  url     = {https://doi.org/10.3389/fphy.2014.00005}
}

@article{GloverKochenbergerDu2019,
  author  = {Glover, Fred and Kochenberger, Gary and Du, Yu},
  title   = {Quantum Bridge Analytics I: A Tutorial on Formulating and Using QUBO Models},
  journal = {4OR},
  year    = {2019},
  volume  = {17},
  pages   = {335--371},
  doi     = {10.1007/s10288-019-00424-y},
  url     = {https://doi.org/10.1007/s10288-019-00424-y}
}

@article{AcerbiTasche2002,
  author  = {Acerbi, Carlo and Tasche, Dirk},
  title   = {On the Coherence of Expected Shortfall},
  journal = {Journal of Banking \& Finance},
  year    = {2002},
  volume  = {26},
  number  = {7},
  pages   = {1487--1503},
  doi     = {10.1016/S0378-4266(02)00283-2},
  url     = {https://doi.org/10.1016/S0378-4266(02)00283-2}
}

@book{McNeilFreyEmbrechts2015,
  author    = {McNeil, Alexander J. and Frey, R{\"u}diger and Embrechts, Paul},
  title     = {Quantitative Risk Management: Concepts, Techniques and Tools},
  edition   = {Revised},
  publisher = {Princeton University Press},
  address   = {Princeton},
  year      = {2015},
  isbn      = {9780691166278},
  url       = {https://press.princeton.edu/books/hardcover/9780691166278/quantitative-risk-management}
}

@misc{FarhiGoldstoneGutmann2014,
  author        = {Farhi, Edward and Goldstone, Jeffrey and Gutmann, Sam},
  title         = {A Quantum Approximate Optimization Algorithm},
  year          = {2014},
  eprint        = {1411.4028},
  archivePrefix = {arXiv},
  primaryClass  = {quant-ph},
  doi           = {10.48550/arXiv.1411.4028},
  url           = {https://arxiv.org/abs/1411.4028}
}

@article{Barron_2024,
   title={Provable bounds for noise-free expectation values computed from noisy samples},
   volume={4},
   ISSN={2662-8457},
   url={http://dx.doi.org/10.1038/s43588-024-00709-1},
   DOI={10.1038/s43588-024-00709-1},
   number={11},
   journal={Nature Computational Science},
   publisher={Springer Science and Business Media LLC},
   author={Barron, Samantha V. and Egger, Daniel J. and Pelofske, Elijah and Bärtschi, Andreas and Eidenbenz, Stephan and Lehmkuehler, Matthis and Woerner, Stefan},
   year={2024},
   month=Nov, pages={865–875} 
}

@misc{ibm_quantum_platform,
  author       = {{IBM Quantum}},
  title        = {IBM Quantum Platform},
  year         = {2025},
  howpublished = {\url{https://quantum.cloud.ibm.com/}},
  note         = {Accessed: 2026-08-28}
}

@manual{cplex,
  author       = {{IBM}},
  title        = {IBM ILOG CPLEX Optimization Studio},
  organization = {IBM},
  year         = {2025},
  url          = {https://www.ibm.com/products/ilog-cplex-optimization-studio}
}

@article{montanez2025lrrqaoa,
  author  = {J. A. Monta{\~n}ez-Barrera and Kristel Michielsen},
  title   = {Toward a Linear-Ramp QAOA Protocol: Evidence of a Scaling Advantage in Solving Some Combinatorial Optimization Problems},
  journal = {npj Quantum Information},
  volume  = {11},
  pages   = {131},
  year    = {2025},
  doi     = {10.1038/s41534-025-01082-1}
}

@article{Angrisani2025NoiselessPauliPropagation,
  author    = {Angrisani, Armando and Schmidhuber, Alexander
               and Rudolph, Manuel S. and Cerezo, M.
               and Holmes, Zo{\"e} and Huang, Hsin-Yuan},
  title     = {Classically Estimating Observables of Noiseless Quantum Circuits},
  journal   = {Physical Review Letters},
  volume    = {135},
  number    = {17},
  pages     = {170602},
  year      = {2025},
  month     = {oct},
  publisher = {American Physical Society},
  doi       = {10.1103/lh6x-7rc3},
  url       = {https://doi.org/10.1103/lh6x-7rc3}
}

@article{Rudolph2026PauliPropagation,
  author    = {Rudolph, Manuel S. and Jones, Tyson and Teng, Yanting
               and Angrisani, Armando and Holmes, Zo{\"e}},
  title     = {Pauli Propagation: A Computational Framework for Simulating
               Quantum Systems},
  journal   = {PRX Quantum},
  volume    = {7},
  pages     = {032001},
  year      = {2026},
  month     = {aug},
  publisher = {American Physical Society},
  doi       = {10.1103/6vd7-l9bn},
  url       = {https://doi.org/10.1103/6vd7-l9bn}
}

@INPROCEEDINGS{riquelme_2015_moo_metrics,
  author={Riquelme, Nery and Von Lücken, Christian and Baran, Benjamin},
  booktitle={2015 Latin American Computing Conference (CLEI)}, 
  title={Performance metrics in multi-objective optimization}, 
  year={2015},
  volume={},
  number={},
  pages={1-11},
  doi={10.1109/CLEI.2015.7360024}
}

@misc{vazquez2026efficientfourierbasedlinearcombination,
      title={Efficient Fourier-Based Linear Combination of Unitaries and Applications in Quantum Optimization}, 
      author={Almudena Carrera Vazquez and Daniel J. Egger and Stefan Woerner},
      year={2026},
      eprint={2605.18985},
      archivePrefix={arXiv},
      primaryClass={quant-ph},
      url={https://arxiv.org/abs/2605.18985}, 
}

@article{Weidenfeller_2022,
   title={Scaling of the quantum approximate optimization algorithm on superconducting qubit based hardware},
   volume={6},
   ISSN={2521-327X},
   url={http://dx.doi.org/10.22331/q-2022-12-07-870},
   DOI={10.22331/q-2022-12-07-870},
   journal={Quantum},
   publisher={Verein zur Forderung des Open Access Publizierens in den Quantenwissenschaften},
   author={Weidenfeller, Johannes and Valor, Lucia C. and Gacon, Julien and Tornow, Caroline and Bello, Luciano and Woerner, Stefan and Egger, Daniel J.},
   year={2022},
   month=Dec, pages={870} 
}

@article{virtanen2020scipy,
  title   = {SciPy 1.0: Fundamental Algorithms for Scientific Computing in Python},
  author  = {Virtanen, Pauli and Gommers, Ralf and Oliphant, Travis E. and
             Haberland, Matt and Reddy, Tyler and Cournapeau, David and
             Burovski, Evgeni and Peterson, Pearu and Weckesser, Warren and
             Bright, Jonathan and van der Walt, St{\'e}fan J. and
             Brett, Matthew and Wilson, Joshua and Millman, K. Jarrod and
             Mayorov, Nikolay and Nelson, Andrew R. J. and Jones, Eric and
             Kern, Robert and Larson, Eric and Carey, C. J. and
             Polat, {\.I}lhan and Feng, Yu and Moore, Eric W. and
             VanderPlas, Jake and Laxalde, Denis and Perktold, Josef and
             Cimrman, Robert and Henriksen, Ian and Quintero, E. A. and
             Harris, Charles R. and Archibald, Anne M. and
             Ribeiro, Ant{\^o}nio H. and Pedregosa, Fabian and
             van Mulbregt, Paul},
  journal = {Nature Methods},
  volume  = {17},
  pages   = {261--272},
  year    = {2020},
  doi     = {10.1038/s41592-019-0686-2}
}

@article{Abbas2024QuantumOptimization,
  author = {Abbas, Amira and Ambainis, Andris and Augustino, Brandon
            and B{\"a}rtschi, Andreas and Buhrman, Harry and Coffrin, Carleton
            and Cortiana, Giorgio and Dunjko, Vedran and Egger, Daniel J.
            and Elmegreen, Bruce G. and Franco, Nicola and Fratini, Filippo
            and Fuller, Bryce and Gacon, Julien and Gonciulea, Constantin
            and Gribling, Sander and Gupta, Swati and Hadfield, Stuart
            and Heese, Raoul and Kircher, Gerhard and Kleinert, Thomas
            and Koch, Thorsten and Korpas, Georgios and Lenk, Steve
            and Marecek, Jakub and Markov, Vanio and Mazzola, Guglielmo
            and Mensa, Stefano and Mohseni, Naeimeh and Nannicini, Giacomo
            and O'Meara, Corey and Pe{\~n}a Tapia, Elena and Pokutta, Sebastian
            and Proissl, Manuel and Rebentrost, Patrick and Sahin, Emre
            and Symons, Benjamin C. B. and Tornow, Sabine and Valls, V{\'i}ctor
            and Woerner, Stefan and Wolf-Bauwens, Mira L. and Yard, Jon
            and Yarkoni, Sheir and Zechiel, Dirk and Zhuk, Sergiy
            and Zoufal, Christa},
  title = {Challenges and Opportunities in Quantum Optimization},
  journal = {Nature Reviews Physics},
  year = {2024},
  volume = {6},
  pages = {718--735},
  doi = {10.1038/s42254-024-00770-9}
}

@article{Koch2026QOBLIB,
  author = {Koch, Thorsten and Bernal Neira, David E. and Chen, Ying
            and Cortiana, Giorgio and Egger, Daniel J. and Heese, Raoul
            and Hegade, Narendra N. and Gomez Cadavid, Alejandro
            and Huang, Rhea and Itoko, Toshinari and Kleinert, Thomas
            and Maciel Xavier, Pedro and Mohseni, Naeimeh
            and Montanez-Barrera, Jhon A. and Nakano, Koji
            and Nannicini, Giacomo and O'Meara, Corey and Pauckert, Justin
            and Proissl, Manuel and Ramesh, Anurag and Schicker, Maximilian
            and Shimada, Noriaki and Takeori, Mitsuharu and Valls, V{\'i}ctor
            and Van Bulck, David and Woerner, Stefan and Zoufal, Christa},
  title = {The Quantum Optimization Benchmarking Library},
  journal = {Nature Computational Science},
  year = {2026},
  volume = {6},
  number = {6},
  pages = {653--671},
  doi = {10.1038/s43588-026-00991-1}
}

@article{elton1973,
  author  = {Elton, Edwin J. and Gruber, Martin J.},
  title   = {Estimating the Dependence Structure of Share Prices---Implications for Portfolio Selection},
  journal = {The Journal of Finance},
  volume  = {28},
  number  = {5},
  pages   = {1203--1232},
  year    = {1973},
  doi     = {10.1111/j.1540-6261.1973.tb01451.x}
}

\end{document}